\documentclass[10pt,reqno]{amsart}
\usepackage{latexsym}

\usepackage{amssymb, amscd, amsthm, multirow}
\usepackage[all]{xy}
\usepackage{verbatim}
\usepackage{bbm}
\usepackage{mathrsfs}

\newtheorem{theorem}{{Theorem}}
\newtheorem{lemma}{{Lemma}}

\newtheorem{proposition}{{Proposition}}

\theoremstyle{remark}
\newtheorem*{remark}{Remark}
\theoremstyle{conjecture}
\newtheorem*{conjecture}{{Conjecture}}

\theoremstyle{definition}
\newtheorem{example}{Example}

\newcommand{\bs}{\mathbf{s}}
\newcommand{\bws}{\widetilde{\mathbf{s}}}
\newcommand{\ws}{\widetilde{s}}

\newcommand{\F}{\mathbb{F}}

\newcommand{\GF}{\mathrm{GF}}
\newcommand{\LC}{\mathrm{LC}}
\newcommand{\Z}{\mathbb{Z}}
\newcommand{\bH}{\mathbf{H}}

\newcommand{\tA}{\widetilde{A}}

\begin{document}
\title[Linear complexity of generalized cyclotomic sequences]{Linear complexity of generalized cyclotomic sequences of period $2p^{m}$}

\author{Yi Ouyang and Xianhong Xie}
\address{Wu Wen-Tsun Key Laboratory of Mathematics,  School of Mathematical Sciences, University of Science and Technology of China, Hefei, Anhui 230026, China}

\email{yiouyang@ustc.edu.cn}
\email{xianhxie@mail.ustc.edu.cn}
\thanks{Partially supported by NSFC No. 11571328}

\begin{abstract} In this paper, we construct two generalized cyclotomic binary sequences of period $2p^{m}$ based on the generalized cyclotomy and compute their linear complexity, showing that they are of high linear complexity when $m\geq 2$.

\smallskip

\end{abstract}
\maketitle

\section{Introduction}
 A sequence $\textbf{s}^{\infty}=\{s_{0},s_{1},s_{2},\ldots\}$ is called a binary sequence of period $N$ if $s_i\in \F_2$ and $s_{i}=s_{i+N}$ for all $i\geq0$. The linear complexity(LC) of a periodic binary sequence $\textbf{s}^{\infty}$,  denoted by $\LC(\textbf{s}^{\infty})$, is the length of shortest linear feedback shift register(LFSR) that generates the sequence (\cite{1}), i.e., the smallest positive integer $l$ such that $s_{i}= c_{l}s_{i-l}+\cdots+c_{2}s_{i-2}+c_{1}s_{i-1}$ for $i\geq l$ and constants $c_{0}=1,c_{1},\ldots,c_{l}\in \F_2$. For $\bs^\infty$  a sequence of period $N$, the
characteristic power series/polynomial of $\textbf{s}^{\infty}$ and $\textbf{s}^{N}=\{s_{0},s_{1},\ldots,s_{N-1}\}$ are defined respectively as $c^{\infty}(x)=s_{0}+s_{1}x +\cdots$ and $c^{N}(x)=s_{0}+s_{1}x+\cdots+s_{N-1}x^{N-1}$, the minimal polynomial (\cite{2})
of $\textbf{s}^\infty$ is
 \[ m(x)=(x^{N}-1)/\gcd(c^{N}(x),x^{N}-1). \]
Then we have the following classical relation
\begin{equation}  \label{eq:LC}
  \LC(\textbf{s}^{\infty})=\deg(m(x))=N-\deg(\gcd(x^{N}-1,c^{N}(x))).
\end{equation}

The linear complexity of a  sequence  is an important criteria of its quality. As we all know, sequences with high linear complexity (such that $\LC(\textbf{s}^{\infty})>\frac{N}{2}$) have important applications in cryptography.

Cyclotomic generators based on cyclotomy can generate sequences with large linear complexity.
Generalized cyclotomic classes with respect to $pq$ and $p^{2}$ were introduced by Whiteman and Ding for the purposes of searching for residue difference sets (\cite{3}) and cryptography (\cite{4}) respectively. Based on Whiteman's generalized cyclotomy of order $2$, Ding (\cite{19}) constructed a class of generalized cyclotomic sequences of period $pq$ and determined their linear complexity. Autocorrelation and linear complexity of period $p^{2}$ and $p^{3}$ were studied in \cite{14,15}. The linear complexity of generalized cyclotomic sequences of period $p^{m}$ were investigated in \cite{16,17}. In addition, the generalized cyclotomy of order $2$ was extended to the case of period $p_{1}^{e_{1}}\cdots p_{m}^{e_{m}}$, which is not consistent with the classical cyclotomy (\cite{5}). Subsequently, new generalized cyclotomic sequences of period $p_{1}^{e_{1}}\cdots p_{m}^{e_{m}}$ that include the classical ones as special cases were presented in \cite{6}, and the linear complexity of such sequences of period $pq$ were calculated in \cite{22}. Furthermore, new classes of generalized cyclotomic sequences of period $2p^{m}$ were proposed in [8], which included the sequence presented in [12] as a special case, and they were shown to have high linear complexity. For recent development of the linear complexity of generalized cyclotomic sequences with different periods, the reader is referred ro \cite{7,8,9,10,13,20}.

In this paper, we construct two new classes of generalized cyclotomic binary sequences of period $2p^{m}$ and compute their linear complexity,  showing that they are  of  high linear complexity when $m\geq 2$.

\subsection*{Acknowledgement} Y. O. would like to thank the Morningside Center of Mathematics for hospitality where part of this paper was written.

\section{Generalized binary cyclotomic sequences of period $2p^{m}$}

Let $p$ be an odd prime and $g$ be a primitive root module $p^{m}$. Replace $g$ by $g+p^m$ if necessary, without loss of generality, we may assume that $g$ is an odd integer, and thus $g$ is a common primitive root module $p^{j}$ and $2p^{j}$ for all $1\leq j\leq m$. For a decomposition $p-1=ef$, write $d_{j}=\frac{\varphi(p^{j})}{e}=p^{j-1}f$ for each $j$ where $\varphi(\cdot)$ is Euler's totient function. For $i\in \Z$, $s=p^j$ or $2p^j$, define
 \begin{equation} D_{i}^{(s)}:=\{g^{i+d_{j}t}\pmod{s}:~0\leq t< e\} =g^{i}D_{0}^{(s)}. \end{equation}
One can see immediately  $D_{i}^{(s)}$ depends only on the congruence class $i\pmod{d_j}$. By abuse of notation we say an integer $n\in D_i^{(s)}$ if $n\pmod{s}\in D_i^{(s)}$.

For $(s,a)=(p^j,p^{m-j})$, $(p^j,2p^{m-j})$ or $(2p^j,p^{m-j})$, we define
 \begin{equation} a  D_{i}^{(s)}:=\{ a g^{i+d_{j}t}\pmod{as}:~0\leq t< e\}. \end{equation}
It is well known that $\{D_{0}^{(p^{j})},D_{1}^{(p^{j})},\ldots,D_{d_{j}-1}^{(p^{j})}\}$ forms a partition of $\Z_{p^{j}}^{*}$ (see \cite{11}), which we call the generalized cyclotomic class of order $d_{j}$ with respect to $p^{j}$, and
\begin{align}
  &\Z_{p^{m}} =\bigcup_{j=1}^{m}\bigcup_{i=0}^{d_{j}-1}p^{m-j}D_{i}^{(p^{j})}\cup \{0\}, \\
   &\Z_{2p^{m}} =\bigcup_{j=1}^{m}\bigcup_{i=0}^{d_{j}-1}p^{m-j}(2D_{i}^{(p^{j})}\cup D_{i}^{(2p^{j})})\cup \{0,p^{m}\}.
\end{align}
From now on, take
 \[ f=2^{r}\ (r\geq1),\ b\in \Z, \ \delta_j=\frac{d_j}{2}=\frac{p^{j-1}f}{2}. \]
In the following we  define two families of generalized cyclotomic sequences of period $2p^{m}$. The ideal of construction stems from Xiao et al. (\cite{18}) for the sequences of period $p^m$.

(i) The generalized cyclotomic binary sequence of period $2p^{m}$ is defined as $\bs^\infty= \{s_i\}_{i\geq 0}$ with
\begin{equation} \label{eq:1}
  s_{i}= \begin{cases}
 {1},\ & \textrm{if}\ i\pmod{2p^{m}}\in C_{1},\\
 {0},\ & \textrm{if}\ i\pmod{2p^{m}}\in C_{0},\\
 \end{cases}
\end{equation}
where
\begin{align}
  &C_{0}=\bigcup_{j=1}^{m}\bigcup_{i=\delta_j}^{d_j-1}
  p^{m-j}(2D_{i+b}^{(p^{j})}\cup D_{i+b}^{(2p^{j})})\cup \{p^{m}\}, \notag \\
  & C_{1}=\bigcup_{j=1}^{m}\bigcup_{i=0}^{\delta_j-1}
  p^{m-j}(2D_{i+b}^{(p^{j})}\cup D_{i+b}^{(2p^{j})})\cup\{0\}\notag.
\end{align}
For the above sequence $\bs^\infty$, the following theorem holds.
\begin{theorem} \label{theorem1}
For the generalized cyclotomic sequence defined by \eqref{eq:1} of period $2p^{m}$,

$(1)$ if $2^{e}\not\equiv \pm 1\pmod{p}$ or $2^e\equiv 1\bmod{p}$ but $2^e\not\equiv 1\bmod{p^2}$, then $\LC(\textbf{s}^{\infty})=2p^{m}$.

$(2)$ if $2^e\equiv -1\bmod{p}$ but $2^e\not\equiv -1\bmod{p^2}$, then $2p^m-2(p-1)\leq \LC(\textbf{s}^{\infty})\leq 2p^{m}-(p-1)$.
\end{theorem}

(ii) The modified generalized cyclotomic binary sequence of period $2p^{m}$ is defined as $\widetilde{\bs}^\infty= \{\widetilde{s}_i\}_{i\geq 0}$ with
\begin{equation} \label{eq:2}
  \widetilde{s}_{i}= \begin{cases}
 {1},\ & \textrm{if}\ i\pmod{2p^{m}}\in \widetilde{C}_{1},\\
 {0},\ & \textrm{if}\ i\pmod{2p^{m}}\in \widetilde{C}_{0}, \\
 \end{cases}
\end{equation}
where
\begin{align}
  &\widetilde{C}_{0}=\bigcup_{j=1}^{m}
  p^{m-j}(\bigcup_{i=0}^{\delta_j-1}2D_{i+b}^{(p^{j})}
  \bigcup_{i=\delta_j}^{d_j-1} D_{i+b}^{(2p^{j})})\cup \{p^{m}\}, \notag\\
  & \widetilde{C}_{1}=\bigcup_{j=1}^{m} p^{m-j}(\bigcup_{i=\delta_j}^{d_j-1}
2D_{i+b}^{(p^{j})}\bigcup_{i=0}^{\delta_j-1}D_{i+b}^{(2p^{j})})
\cup\{0\}\notag.
\end{align}
For the above sequence $\widetilde{\bs}^\infty$, the following theorem holds.
\begin{theorem} \label{theorem2} For the modified generalized cyclotomic sequence defined by \eqref{eq:2} of period $2p^{m}$,

$(1)$ if $2^{e}\not\equiv 1\pmod{p}$, then $\LC(\widetilde{\textbf{s}}^{\infty})=2p^{m}$;

$(2)$ if $2^e\equiv 1\pmod{p}$ but $2^e\not\equiv 1\pmod{p^{2}}$,
then  $2p^m-2(p-1)\leq \LC(\widetilde{\bs}^\infty)\leq 2p^m-(p-1)$.
\end{theorem}
We give several remarks about our main results.
\begin{remark} (1) If $p$ is a non-Weiferich prime (i.e. $2^{p-1}\not\equiv 1\pmod{p^2}$), then our results cover all possible scenarios because $2^e\not\equiv \pm 1\pmod{p^2}$.

(2)  A key argument of the computation follows from the work of Edemskiy et. al.(\cite{23}). Based on our computation,  a new proof of the conjecture by Xiao et al. in \cite{18} can also be achieved.

(3) One would expect similar results hold for any even $f$, not only for $f$ a $2$-power.
\end{remark}

 The inequalities in Theorem~\ref{theorem1}(2) and Theorem~\ref{theorem2}(2), arising from the inseparability of the polynomial $x^{2p^m}-1$ over $\F_2$, are strong enough to deduce that the two generalized sequences are of high linear complexity if $m\geq 2$. For the exact values there, based on numerical evidence, we have the following conjecture:

 \begin{conjecture} $(1)$ If $2^e\equiv -1\bmod{p}$ but $2^e\not\equiv -1\bmod{p^2}$, then $\LC(\textbf{s}^{\infty})= 2p^{m}-(p-1)$.

 $(2)$ If $2^e\equiv 1\pmod{p}$ but $2^e\not\equiv 1\pmod{p^{2}}$,
then  $\LC(\widetilde{\bs}^\infty)= 2p^m-(p-1)-e$.
 \end{conjecture}

\section{Proof of the main results}

Let $\beta=\beta_m$ be a fixed primitive $p^{m}$-th root of unity, which can be considered as an element in  $\GF(2^{n})$ where $n$ is the order of $2$ module $p^{m}$. For $l<m$, $\beta_l=\beta_m^{p^{m-l}}$ is a primitive $p^l$-th root of unity.

We fix the decomposition $p-1=ef$, $f=2^r$ for $r\geq 1$, $\delta_j=\frac{d_j}{2}=\frac{p^{j-1} f}{2}$ for $1\leq j\leq m$ and $b\in \Z$. Note that $\delta_1=\frac{f}{2}$ and $d_1=f$. Set
 \[ \bH_{b}^{(p^{j})}:=\bigcup_{i=0}^{\delta_j-1} p^{m-j}D_{i+b}^{(p^{j})},\quad
 H_{b}^{(p^{j})}:=2\bH_{b}^{(p^{j})}, \quad H_{b}^{(2p^{j})}:=\bigcup_{i=0}^{\delta_j-1}p^{m-j}D_{i+b}^{(2p^{j})} \]
and
 \[ \bH_{b}^{(p^{j})}(x):=\sum_{t\in \bH_{b}^{(p^{j})}} x^t, \ \  H_{b}^{(p^{j})}(x):=\sum_{t\in H_{b}^{(p^{j})}} x^t=\bH_{b}^{(p^{j})}(x^2),\ \ H_{b}^{(2p^{j})}(x):=\sum_{t\in H_{b}^{(2p^{j})}} x^t. \]
Set
\begin{align*}
  &s(x):=\sum_{t\in C_{1}}x^{t}=1+\sum_{j=1}^{m}(H_{b}^{(p^{j})}(x)+H_{b}^{(2p^{j})}(x)),\\
  &\widetilde{s}(x):=\sum_{t\in \widetilde{C}_{1}}x^{t}=1+\sum_{j=1}^{m}({H}_{b+\delta_j}^{(p^{j})}(x)+
  H_{b}^{(2p^{j})}(x)).
\end{align*}
To study the linear complexity of $\textbf{s}^{\infty}$ and $\widetilde{\textbf{s}}^{\infty}$, note that there is some subtlety here: the polynomial $x^{2p^m}-1$ is inseparable, each root $\beta^a$ ($a\in \Z_{p^m}$) is of multiplicity $2$, so by \eqref{eq:LC}, we have the inequalities
 \begin{align} & 2p^m-2|\{a\in \Z_{p^m}\mid s(\beta^a)=0\}|\leq \LC(\textbf{s}^{\infty})\leq 2p^m-|\{a\in \Z_{p^m}\mid s(\beta^a)=0\}|,\\
 & 2p^m-2|\{a\in \Z_{p^m}\mid \ws(\beta^a)=0\}|\leq \LC(\bws^{\infty})\leq 2p^m-|\{a\in \Z_{p^m}\mid \ws(\beta^a)=0\}|. \end{align}
Since the polynomial is valued over a field of characteristic $2$, we have
 \begin{align}
 & H_{b}^{(p^{j})}(\beta^{a})=\bH_{b}^{(p^{j})}(\beta^{2a})
 =(\bH_{b}^{(p^{j})}(\beta^a))^2, \\
 & H_{b}^{(2p^{j})}(\beta^{a})=\bH_{b}^{(p^{j})}(\beta^a).
   \end{align}
To study $s(\beta^a)$ and $\ws(\beta^a)$, it suffices to evaluate $\bH_{b}^{(p^{j})}(\beta^a)$  for each $j\leq m$.

\begin{lemma}[\cite{18}, Lemma 4] \label{lemma:4} We have
 \begin{align}  &\bH_{b}^{(p)}(\beta)+\bH_{b+\frac{f}{2}}^{(p)}(\beta)=\sum_{t\in p^{m-1}\Z_{p}^{*}}\beta^{t}=1,\\
 & \bH_{b}^{(p^j)}(\beta)+\bH_{b+\delta_j}^{(p^j)}(\beta)=\sum_{t\in p^{m-j}\Z_{p^j}^{*}}\beta^{t}=0 \ \text{if}\ 2\leq j\leq m.
 \end{align}
\end{lemma}

\begin{lemma} \label{lemma:2}   Let $a=p^{l}u\in p^{l}D_{k}^{(p^{m-l})}$  where $0\leq l\leq m-1$. Then for $j=1,2,\cdots,m$,
\begin{enumerate}
 \item if $j\leq l$, $\bH^{(p^{j})}_{b}(\beta^{a}) = \frac{p^{j-1}(p-1)}{2}$;
\item if $j=l+1$, $\bH^{(p^{j})}_{b}(\beta^{a}) =
 {\frac{p^{l}-1}{2}+\bH_{b+k}^{(p)}(\beta)}$;
 \item if $j>l+1$, $\bH^{(p^{j})}_{b}(\beta^{a}) =
 \bH_{b+k}^{(p^{j-l})}(\beta)$.
\end{enumerate}
\end{lemma}

\begin{proof} First note the computation here is carried out in $\GF(2^{n})$. By definition,
\begin{equation}\label{eq:5}
  \bH_{b}^{(p^{j})}(\beta^{a})=\sum_{t\in H_{b}^{(p^{j})}}\beta^{at}=\sum_{i=0}^{\delta_j-1}\sum_{t\in p^{m-j}D_{i+b}^{(p^{j})}}\beta^{tp^{l}u}=\sum_{i=0}^{\delta_j-1}\sum_{t\in p^{m+l-j}D_{i+b}^{(p^{j})}}\beta^{tu}.
\end{equation}

If $j\leq l$, each term in $ \bH_{b}^{(p^{j})}(\beta^{a})$ defined in \eqref{eq:5} equals to $1$, hence
 \[ \bH_{b}^{(p^{j})}(\beta^{a})=\delta_j\mid D_{i+b}^{(p^{j})}\mid=\delta_jp^{j-1}\frac{p-1}{p^{j-1}f}
 = \frac{p^{j-1}(p-1)}{2}. \]

If $j>l$, let $s=j-l$, then
\begin{equation}\label{eq:6}
  \bH_{b}^{(p^{j})}(\beta^{a})=\sum_{i=0}^{\delta_j-1}\sum_{t\in p^{m+l-j}D_{i+b}^{(p^{j})}}\beta^{tu}=\sum_{i=0}^{\delta_j-1}\sum_{t\in D_{i+b}^{(p^{j})}}\beta^{p^{m-s}tu}.
\end{equation}
Note that when $i$ passes through $\{0,1,\ldots,\delta_j-1\}$, $i\pmod{d_s}$ takes value $\frac{p^{l}-1}{2}$ times on each element in $\{0,1,\ldots,d_s-1\}$ and one additional time on elements in $\{0, 1,\ldots, \delta_s-1\}$. Hence the multiset
 \[ \{tu\pmod{p^s}\mid t\in D_{i+b}^{(p^{j})},\ 0\leq i\leq \delta_j-1\} \]
 passes $\frac{p^l-1}{2}$ times through $\Z_{p^s}^*$, and  one additional time over the union of $ D_{i+k+b}^{(p^s)}$ for $0\leq i\leq \delta_s-1$. Since $\beta^{p^{m-s}}$ is a primitive $p^s$-th root of unity, by \eqref{eq:6}, we have
 \[
  \bH_{b}^{(p^{l+1})}(\beta^{a})
   =\frac{p^l-1}{2}  \sum_{a\in \Z_{p^s}^*} \beta^{p^{m-s}a}+\bH_{b+k}^{(p^s)}(\beta),
 \]
which is $\frac{p^l-1}{2} +\bH_{b+k}^{(p)}(\beta)$ if $s=1$ and $\bH_{b+k}^{(p^s)}(\beta)$ if $s\geq 2$ by Lemmas~\ref{lemma:4}.
\end{proof}

From Lemma~\ref{lemma:2}, we  have the following easy consequence.

\begin{proposition} \label{prop:1} For $a=0$, one has $ s(1)=\widetilde{s}(1)=1$. For $a\in p^{l}D_k^{(p^{m-l})}$, $0\leq l<m$, one has
 \begin{align}
 & s(\beta^{a}) =
 {1+\sum_{s=1}^{m-l}\bH_{b+k}^{(p^s)}(\beta)+\sum_{s=1}^{m-l}
 \bH_{b+k}^{(p^s)}(\beta)^2},\\
 & \widetilde{s}(\beta^{a}) =
 {\sum_{s=1}^{m-l}{\bH}_{b+k}^{(p^s)}(\beta)+
 \sum_{s=1}^{m-l}\bH_{b+k}^{(p^s)}(\beta)^2}.
\end{align}
\end{proposition}
Note that $\bH_v^{(p^s)}=\bH_{m,v}^{(p^s)}$ depends on $m$, however, for $s\leq l\leq m$, $\bH_{m,v}^{(p^s)}(\beta)=\bH_{l,v}^{(p^s)}(\beta_l)$. Thus for the sum  $\sum\limits_{s=1}^{m-l} \bH_v^{(p^s)}(\beta)$ appeared in Proposition~\ref{prop:1}, we may assume $m-l$ is just $m$. For $v\in \Z$,
write
 \[ A_{l,v}(x)=\sum\limits_{s=1}^{1} \bH_v^{(p^s)}(x)\quad
 \text{and}\quad A_{l,v}=A_{1,v}(\beta). \]
We drop the subscript $l$ if $l=m$. By Lemma~\ref{lemma:4}, we have
 \begin{equation} A_{l,v}+A_{l,v+\delta_l}=1.
 \end{equation}

\begin{proposition} \label{prop:2} Suppose $2\in D_h^{(p^m)}$.

$(1)$ If $2\in D_0^{(p^m)}$, then $A_{v}\in \F_2$ for $m\geq 1$. If $2\notin D_0^{(p)}$, then $A_v\notin \F_2$ for $m\geq 1$.

$(2)$ If $2\in D_0^{(p)}$ but $2\notin D_0^{(p^2)}$, then $A_{1,v}\in \F_2$ and $A_{v}\notin \F_4$ for  $m\geq 2$.

$(3)$ If $2\in D_{\delta_1}^{(p)}$ but $2\notin D_{\delta_2}^{(p^2)}$(note that $\delta_1=\frac{f}{2}$ and $\delta_2=p\delta_1$), then $A_{1,v}\in \F_4- \F_2$ and $A_{v}\notin \F_4$ for $m\geq 2$.

$(4)$ If $\delta_1\nmid h$, then $A_{v}\notin \F_4$ for any $v$.
\end{proposition}
\begin{proof} If $2\in D_h^{(p^m)}$, then $2\in D_h^{(p^s)}$ for all $s\leq m$. For any $i$, we have $\{2t\mid t\in D_i^{(p^s)}\}=D_{i+h}^{(p^s)}$, hence $\bH_{v}^{(p)}(\beta)^2=\bH_{v}^{(p)}(\beta^2)=\bH_{v+h}^{(p)}(\beta)$
and
 \begin{equation}\label{eq:key} A_{v}^2=A_{v+h}.
 \end{equation}

(1) If $2\in D_{0}^{(p^m)}$, then \eqref{eq:key} implies that $A_{v}^2=A_{v}$, hence $A_v\in \F_2$.

If $2\notin D_{0}^{(p)}$, then $f\nmid h$, there exists $x_{1}>0$ such that $hx_1\equiv \delta_m\pmod{d_m}$. By Lemma~\ref{lemma:4}, we have
\begin{equation*}
  A_{v+hx_1}=A_{v+\delta_m}=A_{v}+1.
\end{equation*}
On the other hand, if $A_{v}\in\F_2$, by \eqref{eq:key}, for all $n\in \Z$, we have
 \[ A_{v}=A_{v\pm h}=\cdots= A_{v+nh} \in \F_2. \]
This is a contradiction.

(2) That $A_{1,v}\in \F_2$ follows from (1). For $m\geq 2$, the assumption means $\gcd(h, d_m)=d_1=f$ and hence $\gcd(h, \delta_m)=\delta_1$. Let $\tA_{w}(x)=A_{w}(x)-A_{1,w}(x)$ and $\tA_{w}=A_{w}-A_{1,w}$. If $A_v\in \F_2$, then for $n\in\Z$,
  \[ \tA_{v}=\tA_{v\pm h}=\cdots= \tA_{v+nh} \in \F_2, \]
and by Lemma~\ref{lemma:4},
 \[ \tA_{v}=\tA_{v\pm \delta_m}=\cdots= \tA_{v+n\delta_m} . \]
Hence $\tA_v=\tA_{v+n_1 h+n_2\delta_m}$ for any $n_1,\ n_2\in \Z$, and
$\tA_{v}=\tA_{v+n\delta_1}$ for $n\in\Z$.

Note that the condition $2\in D_0^{(p^j)}$ is nothing but $2^e\equiv 1\pmod{p^j}$. If we let the order of $2$ modulo $p$ be $\tau$, then our assumption means that $\tau\mid e$ and the order of $2$ modulo $p^j$ is $\tau p^{j-1}$ for $j\geq 2$ by the argument in \cite[Lemma 2]{23}, and hence $[\F_2(\beta):\F_2(\beta_{m-1})]=p$ by \cite[Lemma 3]{23}. Then essentially the same argument in \cite[Proposition 2]{23} shows that
$\tA_{v}=\tA_{v+n\delta_1}$ is not possible. Hence $A_v\notin \F_2$.

If $A_v\in \F_4-\F_2$, then $\tA_v\in \F_4-\F_2$, we have $\tA_{v+h}=\tA^2_v=\tA_v+1$ and $\tA_{v+2h}=\tA_v$; and $\tA^2_{v-h}=\tA_v=(\tA_v+1)^2$, $\tA_{v-h}=\tA_v+1$ and $\tA_{v-2h}=\tA_v$. Hence we still have $\tA_{v}=\tA_{v+n\delta_1}$ for $n\in \Z$. Now apply the argument in the previous paragraph again to get a contradiction.

(3) Since $2\in D_{\delta_1}^{(p)}$. Hence
 \[ A^2_{1,v}=A_{1,v+\delta_1}=A_{1,v}+1 \]
and $A_{1,v}\in \F_4$. For $m>1$, again let $\tA_v=A_v-A_{1,v}$. Then $\tA_v^2=\tA_{v+h}$. If $\tA_v\in \F_2$, we have $\tA_{v+h}=\tA_v$,  If $\tA_v\in \F_4-\F_2$, we have $\tA_{v\pm 2h}=\tA_v$. Since by assumption, $\gcd(h,\delta_m)=\gcd(2h,\delta_m)=\delta_1$, we  get $\tA_{v}=\tA_{v+n\delta_1}$. Now $2\in D_{\delta_j}^{(p^j)}$ is nothing but the condition $2^e\equiv -1\bmod{p^j}$. Our assumption still leads to the condition $[\F_2(\beta):\F_2(\beta_{m-1})]=2$, and we can still follow the argument in \cite[Proposition 2]{23} to get a contradiction.

(4) If $\frac{f}{2}\nmid h$, then in particular $\frac{f}{2}=2^{r-1}$ is even and there exists an even integer $x_{1}>0$ such that $hx_1\equiv \frac{f}{2}\pmod{f}$. If $A_{v}\in\F_4$, by the proof of (1), we may assume $A_{v}=\epsilon_0\notin \F_2$, thus $\epsilon_0^2+\epsilon_0+1=0$. By
Lemma~\ref{lemma:4}, we have
\begin{equation*}
  \epsilon_{p^{l-1}x_1}:=A_{v+p^{m-1}x_1}=A_{v+\delta_m}
  =A_{v}+1=\epsilon_0+1.
\end{equation*}
Then  by \eqref{eq:key}, we have
$\epsilon_1=A_{v+h}=\epsilon_0^2=\epsilon_0+1$, $\epsilon_2=A_{v+2h}=\epsilon_1^2=\epsilon_0$, hence $\epsilon_0=\epsilon_2=\cdots=\epsilon_{p^{m-1}x_1}$. This is a contradiction.
\end{proof}
\begin{remark} It is well known that
 \[ 2^e\equiv 1\pmod{p^j}\Longleftrightarrow 2\in D_0^{(p^j)},\ \text{and}\ 2^e\equiv -1\pmod{p^j}\Longleftrightarrow 2\in D_{\delta_j}^{(p^j)}. \]

(1) If $p$ is a non-Wieferich prime, i.e., $2^{p-1}\not\equiv 1\bmod{p^2}$, then it is always true that $2\notin D^{(p^2)}_h$ for $h=0$ or $\delta_2$. So the proposition covers all possible $h$ that  $2\in D_h^{(p^m)}$. Consequently, the conjecture by Xiao et al. in \cite{18} can be proved.

(2) Suppose $p$ is a Wieferich prime.  Suppose $2\in D_0^{(p^a)}$ but $2\notin D_0^{(p^a+1)}$ for some $1\leq a<m$, then $A_{s,v}\in \F_2$ for $s\leq a$. We also have $[\F_2(\beta_{j+1}):\F_2(\beta_j)]=p$ for $j\geq p$. If $A_v\in \F_4$, write $\tA_v=A_v-A_{a,v}$, then we can get $\tA_v=\tA_{v+n\delta_a}$ for $n\in \Z$. It would be great if we can get a contradiction, however, the argument in \cite{23} only works for the case $a=1$.

Similarly,  suppose $2\in D_{\delta_a}^{(p^a)}$ but $2\notin D_{\delta_{a+1}}^{(p^a+1)}$ for some $1\leq a<m$, then $A_{s,v}\in \F_4$ for $s\leq a$. We also have $[\F_2(\beta_{j+1}):\F_2(\beta_j)]=p$ for $j\geq p$. If $A_v\in \F_4-\F_2$, write $\tA_v=A_v-A_{a,v}$, then we also get $\tA_v=\tA_{v+n\delta_a}$ for $n\in \Z$.
\end{remark}

We are now ready to prove our main results by applying Propositions~\ref{prop:1} and \ref{prop:2}.

\begin{proof}[Proof of Theorem~\ref{theorem1}] If $2^e\equiv 1\pmod{p}$ but $2^e\not\equiv 1\pmod{p^2}$, then $A_{1,v}\in \F_2$ and $A_{l,v}\notin \F_4$ for $l\geq 2$, in both cases,  $s(\beta^a)=1\neq 0$. If $ 2^e\not\equiv \pm1\pmod{p}$, then $\delta_1\nmid h$ and $A_{l,v}\notin \F_4$, and hence $s(\beta^a)\neq 0$. Therefore  $\LC(\bs^\infty)=2p^m$.

If $2^e\equiv -1\mod{p}$ but  $2^e\not\equiv -1\pmod{p^2}$, then $A_{1,v}\in \F_4-\F_2$ and $A_{l,v}\notin \F_4$ for $l\geq 2$. Hence
$s(\beta^a)=0$ for $a\in p^{m-1}\Z_p^*$ and $s(\beta^a)\neq 0$ for
all other $a$. Hence $2p^m-2(p-1)\leq \LC(\bs^\infty)\leq 2p^m-(p-1)$.
\end{proof}

\begin{proof}[Proof of Theorem~\ref{theorem2}]  If $2^e\not\equiv 1\pmod{p}$, then $2\notin D_0^{(p)}$. Hence  $A_{l,v}\notin \F_2$ for all $l$ and $\widetilde{s}(\beta^a)\neq 0$. Therefore  $\LC(\widetilde{\bs}^\infty)=2p^m$.

If  $2^e\equiv 1\pmod{p}$ but $2^e\not\equiv 1\pmod{p^{2}}$, then only $A_{1,v}\in \F_2$ and $\widetilde{s}(\beta^a)= 0$ for $a\in p^{m-1} \Z_{p}^*$. For all other $a$, $\widetilde{s}(\beta^a)\neq 0$. Hence $2p^m-2(p-1)\leq \LC(\widetilde{\bs}^\infty)\leq 2p^m-(p-1)$.
\end{proof}

\section{Numerical Evidence}
By using Magma, we compute the following examples to check our results.

\begin{example} Let $p=7$, $m=2$ and $g=3$. Take $f=2$ and $e=3$, then $2^3\equiv1\pmod{p}$ and $2^3\not\equiv1\pmod{p^2}$. For $b=0$,
 \[ \begin{split} \bs^\infty=&111101110110011100100000011111101
 0001101010101010\\ &01010101010100111
 01000000111111011000110010001000.
 \end{split} \]
Then $\LC(\bs^\infty)=98=2p^m$, consistent with Theorem~\ref{theorem1}(1). For $b=0$,
 \[ \begin{split} \bws^\infty=&110111011100110110001010110101000
 0100111111111111\\ & 00000000000001101
 11101010010101110010011000100010.
 \end{split} \]
Then $\LC(\bws^\infty)=89=2p^m-(p-1)-e$, consistent with Theorem~\ref{theorem2}(2).
 \end{example}

\begin{example} Let $p=5$, $m=2$ and $g=3$. Then  $f$ can be taken either $2$ or $4$.

(i) If one takes $f=2$, then $e=2$, $2^2\equiv-1\pmod{p}$  and $2^2\not\equiv-1\pmod{p^2}$. For $b=0$,
 \[  \bs^\infty=1111111001101000001100010
 0010001100000101100111111.
  \]
Then $\LC(\bs^\infty)=46=2p^m-(p-1)$, consistent with Theorem~\ref{theorem1}(2). For $b=0$,
 \[  \bws^\infty=1101010011000010110110111
 0111011000010000110010101.
\]
Then $\LC(\bws^\infty)=50=2p^m$, consistent with Theorem~\ref{theorem2}(1).

(ii) If one takes $f=4$, then $e=1$, $2\not\equiv1\pmod{p}$. For $b=0$,
 \[ \begin{split} & \bs^\infty=1111111011111001101000101
 0010111010011000001000000, \\
 &\widetilde{\bs}^\infty=1101010001010011000010000
 0111101111001101011101010.
 \end{split} \]
Then $\LC(\bs^\infty)=\LC(\bws^\infty)=50=2p^m$, consistent with Theorem~\ref{theorem1}(1) and Theorem~\ref{theorem2}(1) respectively.\end{example}

\begin{example} Let $p=31$, $m=1$, $g=3$ and $e=15$. Then $2^{15}\equiv1\pmod{31}$ and $2^{15}\not\equiv1\pmod{31^2}$. For $b=0$,
 \[ \begin{split}
 & \bs^\infty=1110110111100010101110000100100
 0110110111100010101110000100100,\\ &\widetilde{\bs}^\infty=1100011101001000000100101110001
 0011100010110111111011010001110.\\
 \end{split} \]
Then $\LC(\bs^\infty)=62=2p$ and $\LC(\widetilde{\bs}^\infty)=17=2p-(p-1)-e$, consistent with Theorem~\ref{theorem1}(1) and Theorem~\ref{theorem2}(2). \end{example}

Because of the above examples, we form our conjecture and try more examples in the following two tables:

\begin{table}[!htbp]
\caption{$\LC(\bs^\infty)$ for $2^e\equiv -1\bmod{p}$ but $\not\equiv -1\pmod{p^2}$}
 \begin{tabular}{|c|c|c|c|c|c|c|}
     \hline
       $p$ & $m$ & $e$ & $g$ & $b$ & $\LC(\bs^\infty)$ & $2p^m-(p-1)$\\
       \hline
       \multirow{3}{*}{5} & 2 & \multirow{3}{*}{2} &  \multirow{3}{*}{3}  & \multirow{3}{*}{0, 1, 3}
        & 46 & 46 \\ \cline{2-2} \cline{6-7}
       & 3 & & &  & 246 & 246 \\ \cline{2-2} \cline{6-7} & 4 & & &  & 1246 & 1246\\
       \hline
       11 & 2 & 5 & 7 & 2, 19 & 232 & 232 \\ \cline{1-7}
       \multirow{4}{*}{13} & \multirow{2}{*}{2} & \multirow{4}{*}{6} & 7 & {6, 11} &  \multirow{2}{*}{326} & \multirow{2}{*}{326}\\ \cline{4-5} & & & 11 & 5, 12 & & \\ \cline{2-2} \cline{4-7} & \multirow{2}{*}{3} & & 7 & \multirow{2}{*}{5, 12} & \multirow{2}{*}{4382} & \multirow{2}{*}{4382}\\ \cline{4-4} & & & 11 & & &\\ \hline
       \multirow{4}{*}{17} & \multirow{2}{*}{1} & \multirow{4}{*}{4} & 3 & \multirow{2}{*}{0, 3} &  \multirow{2}{*}{18} & \multirow{2}{*}{18}\\ \cline{4-4} & & & 5 &  & & \\ \cline{2-2} \cline{4-7} & \multirow{2}{*}{2} & & 3 & {0, 2} & \multirow{2}{*}{562} & \multirow{2}{*}{562}\\ \cline{4-5} & & & 5 & 0, 7 & &\\ \hline
        \multirow{2}{*}{19} & \multirow{2}{*}{2} & \multirow{2}{*}{9} & 3 & {1, 6} &  \multirow{2}{*}{704} & \multirow{2}{*}{704}\\ \cline{4-5} & & & 13 & 3, 22 & & \\ \hline
   \end{tabular}
\end{table}

\begin{table}[!htbp]
\caption{$\LC(\widetilde{\bs}^\infty)$ for $2^e\equiv 1\bmod{p}$ but $\not\equiv 1\pmod{p^2}$}
 \begin{tabular}{|c|c|c|c|c|c|c|}
     \hline
       $p$ & $m$ & $e$ & $g$ & $b$ & $\LC(\widetilde{\bs}^\infty)$ & $2p^m-(p-1)-e$\\
       \hline
        \multirow{4}{*}{7} & \multirow{2}{*}{2} & \multirow{4}{*}{3} & 3 & \multirow{2}{*}{0, 3} &  \multirow{2}{*}{89} & \multirow{2}{*}{89}\\ \cline{4-4} & & & 5 &  & & \\ \cline{2-2} \cline{4-7} & \multirow{2}{*}{3} & & 3 & \multirow{2}{*}{0, 1} & \multirow{2}{*}{677} & \multirow{2}{*}{677}\\ \cline{4-4} & & & 5 & & &\\
       \hline
       \multirow{2}{*}{17} & {1} & \multirow{2}{*}{8} & 3 & {0, 3} &  {10} & {10}\\  \cline{4-4}   \cline{2-2} \cline{4-7} & {2} & & 5 & {0, 3} & {554} & {554}\\  \hline
       \multirow{2}{*}{23} & \multirow{2}{*}{2} & \multirow{2}{*}{11} & 5 & \multirow{2}{*}{1, 5} &  \multirow{2}{*}{1025} & \multirow{2}{*}{1025}\\ \cline{4-4} & & & 7 &  & & \\ \hline
       \multirow{4}{*}{31} & \multirow{2}{*}{1} & \multirow{4}{*}{15} & 3 & \multirow{2}{*}{1, 6} &  \multirow{2}{*}{17} & \multirow{2}{*}{17}\\ \cline{4-4} & & & 11 &  & & \\ \cline{2-2} \cline{4-7} & \multirow{2}{*}{2} & & 3 & \multirow{2}{*}{2, 5} & \multirow{2}{*}{1877} & \multirow{2}{*}{1877}\\ \cline{4-4} & & & 11 & & &\\ \hline
   \end{tabular}
\end{table}

\section{Conclusion}
In this paper, we introduced two generalized cyclotomic binary sequences of period $2p^{m}$, which include the sequences in \cite {8, 12} as special cases. We computed their linear complexity in most cases (all cases for $p$ a non-Wieferich odd prime) and showed each of our sequences is of high linear complexity if $m\geq 2$.

\end{document}